\documentclass[conference]{IEEEtran}
\IEEEoverridecommandlockouts

\usepackage{graphicx}

\usepackage[cmex10]{amsmath}
\usepackage{amsmath}
\usepackage{amssymb}
\usepackage{amsthm}
\usepackage{verbatim}
\usepackage{color}
\usepackage{cite}
\usepackage{graphicx,subcaption}

\ifCLASSINFOpdf
\else
\fi

\newtheorem{definition}{Definition}

\newtheorem{proposition}{Proposition}
\newtheorem{theorem}{Theorem}

\newtheorem{example}{Example}

\begin{document}

\def\N{\mathcal{N}}
\def\As{A_{\mathbf{s}}}
\def\Bs{B_{\mathbf{s}}}
\def\Cs{C_{\mathbf{s}}}
\def\SN{S_{\mathcal{N}}}
\def\w{\mathbf{w}}
\def\z{\mathbf{z}}
\def\c{\mathbf{c}}
\def\u{\mathbf{u}}
\def\h{\mathbf{h}}
\def\g{\mathbf{g}}
\def\f{\mathbf{f}}
\def\0{\mathbf{0}}
\def\s{\mathbf{s}}
\def\Aw{A_{\mathbf{w}}}
\def\Bw{B_{\mathbf{w}}}
\def\Cw{C_{\mathbf{w}}}
\def\M{\mathcal{M}}
\def\C{\mathcal{C}}
\def\I{\mathcal{I}}
\def\Cd{\mathcal{C}^\perp}
\def\FF{\mathbb{F}}
\def\ZZ{\mathbb{Z}}
\def\EE{\mathbb{E}}
\def\FFq{\mathbb{F}_q}
\def\FFqn{\mathbb{F}_{q}^{N}}
\def\FFqk{\mathbb{F}_{q}^{k}}
\def\lambdaC{\Lambda_{\mathcal{C}}}
\def\lambdaCd{\Lambda_{\mathcal{C}^{\perp}}}

\def\support{{\Lambda}}
\def\br{{\bf r}}
\def\bc{{\bf c}}
\def\bz{{\bf z}}
\def\bs{{\bf s}}
\def\bu{{\bf u}}
\def\bv{{\bf v}}
\def\bg{{\bf g}}
\def\bff{{\bf f}}
\def\ground{{\Omega}}
\def\bc{{\bf c}}
\def\solA{{\bf A}}
\def\solB{{\bf B}}
\def\solC{{\bf C}}
\def\field{{\mathbb F}_{q}}
\def\Kp{{\mathcal K}}

\def\N{{\cal N}}
\def\bz{{z_{\scriptscriptstyle \N}}}
\def\by{{y_{\scriptscriptstyle \N}}}
\def\l{\left}
\def\r{\right}
\def\gf{ {\mathbb F}}
\def\rankfn{\rho}
\def\A{\mathcal A}
\def\M{\mathcal M}
\def\X{\mathcal X}
\def\B{\mathcal B}
\def\reals{\mathbb R}
\def\Z{\mathcal Z}
\def\K{\mathcal K}
\def\C{\mathcal C}
\def\p{\prime}
\def\real{{\mathbb R}}
\def\dist{{W}}
\def\T{\mathcal T}

\def\x{{\bf x}}
\def\y{{\bf y}}

\def\pattern{{\gamma}}
\def\erasurematrix{{\bf I}}

\def\X{{\bf X}}
\def\H{{\bf H}}
\def\solG{{\bf G}}
\def\zero{{\bf 0}}
\def\r{{\bf r}}

\def\R{\mathcal R}
\def\P{\mathcal P}
\def\cbP{\mathbf P}
\def\FN{\mathbf {FailedNodes}}
\def\FP{\mathbf {FailurePatterns}}
\def\A{\mathcal{A}}
\def\F{\mathcal{F}}
\def\G{\mathcal{G}}
\def\S{\mathcal{S}}
\def\J{\mathcal{J}}
\def\E{\mathcal{E}}
\def\HH{\mathcal{H}}
\def\j{\mathbf{J}}
\def\L{\mathcal{L}}
\def\U{\mathcal{U}}
\def\q{\mathbf{q}}
\def\v{\mathbf{v}}
\def\D{\mathcal{D}}
\def\Q{\mathcal{Q}}
\def\d{\mathbf{d}}
\def\k{\mathbf{k}}
\def\b{\mathbf{b}}
\def\solS{\mathbf{S}}
\def\solF{\mathbf{F}}

\def\vecd{\underline{d}}

\def\lmd{{{\textsf {lmd}}}}
\def\gmd{{{\textsf {gmd}}}}

\def\sK{{{\textsf {K}}}}
\def\sf{{{\textsf {f}}}}
\def\sA{{{\textsf {A}}}}
\def\sU{{{\textsf {U}}}}
\def\sr{{{\textsf {r}}}}
\def\su{{{\textsf {u}}}}
\def\sT{{{\textsf {T}}}}
\def\st{{{\textsf {t}}}}
\def\sR{{{\textsf {R}}}}

\def\G{{G}}
\def\bx{{\bf x}}
\def\by{{\bf y}}
\def\bg{{\bf g}}
\def\zeroz{\textbf{0}}
\def\X{{X}}

\long\def\/*#1*/{}

% paper title
\title{Coded Caching in Partially Cooperative D2D Communication Networks%Coded Caching for D2D Communication Networks with Selfish Users
\thanks{This work was partially supported by a grant from the University Grants Committee of the Hong Kong Special Administrative Region,
China (Project No. AoE/E-02/08).}}
% Authors
\author{\IEEEauthorblockN{Ali Tebbi, Chi Wan Sung}
\IEEEauthorblockA{\textit{Department of Electronic Engineering}\\ 
\textit{City University of Hong Kong}\\
Email: ali.tebbi, albert.sung@cityu.edu.hk}}

\begin{comment}
% Authors
\author{\IEEEauthorblockN{M. Ali Tebbi}
\IEEEauthorblockA{Institute for Telecommunications\\ Research,\\University of South Australia\\
mohammad$\_$ali.tebbi@mymail.unisa.edu.au}	
\and
\IEEEauthorblockN{Terence H. Chan}
\IEEEauthorblockA{Institute for Telecommunications\\ Research,\\University of South Australia\\
terence.chan@unisa.edu.au}	
\and
\IEEEauthorblockN{Chi Wan Sung}
\IEEEauthorblockA{Department of Electronic Engineering,\\
City University of Hong Kong\\
albert.sung@cityu.edu.hk}}
\end{comment}

% use for special paper notices
%\IEEEspecialpapernotice{(Invited Paper)}

% make the title area
\maketitle

\begin{abstract}
%\boldmath
The backhaul traffic is becoming a major concern in wireless and cellular networks (e.g., 4G-LTE and 5G) with the increasing demands for online video streaming. Caching the popular content in the cache memory of the network users (e.g., mobile devices) is an effective technique to reduce the traffic during the networks' peak time. However, due to the dynamic nature of these networks, users privacy settings, or energy limitations, some users may not be available or intend to participate during the caching procedures. In this paper, we propose caching schemes for device-to-device communication networks where a group of users show selfish characteristics. The selfish users along with the non-selfish users will cache the popular content, but will not share their useful cache content with the other users to satisfy a user request. We show that our proposed schemes are able to satisfy any arbitrary user requests under partial cooperation of the network users.
\end{abstract}


\IEEEpeerreviewmaketitle

%%%%%%%%%%%%%%%%%%%%%%%%%%%%%%%%%%%%%%%%%%%%%%%%%%%%%%
\section{Introduction}
 % no \IEEEPARstart
%%%%%%%%%%%%%%%%%%%%%%%%%%%%%%%%%%%%%%%%%%%%%%%%%%%%%%
The data traffic in wireless and mobile networks is exponentially growing during the recent years specially due to the expanding demands of multimedia contents such as YouTube videos \cite{CISCO2017}. The massive resultant backhaul traffic in wireless networks can be reduced by bringing the popular content forward closer to the edge of the network. This technique is called ``caching`` which basically is storing the popular content in resources closer to the end users such as base stations, routers, and mobile devices in order to serve the users' requests. A caching scheme generally consists of two phases: caching (or placement) phase and delivery phase. The caching phase carries out during the off-peak time of the network when the network is not congested to duplicate and distribute the popular content across the network. The delivery phase is carried out during the peak time of the network when the network is congested in order to serve the users' requests by the contents cached in the network.

\begin{comment}
Conventional caching schemes operations are based on optimizing the caching phase and delivery phase separately \cite{FileAssign,ApproxAlgPlace,ISCOD} where the caching phase is optimized for a fixed delivery procedure or the delivery phase is optimized for a fixed cache content and user request.
\end{comment}

A novel approach of caching is recently introduced by Maddah-Ali $\textit{et al.}$ in \cite{CodedCaching} known as ``coded caching`` where it aims to jointly optimize the caching and delivery phase in order to achieve a global caching gain (i.e., caching gain which is attained from aggregate global cache size). 
%In this approach, $K$ users are connected to a server (or base station) through a shared link while a set of $N$ popular files are stored in the server. 
The cache placement in this approach is based on constructing a multicast opportunity for all users at the same time. In other words, the server can deliver any arbitrary users demands by a number of coded multicast transmissions. %This simultaneous multicasting opportunity is created by the cache placement policy such that any subset of users' cache memories share a specific part of each file.

Since the introduction of the coded caching, this approach received a considerable attention due to its noticeable caching gain (i.e., achieving local and global gain at the same time) over the conventional schemes. The centralized nature of the caching scheme in \cite{CodedCaching} is quite far from the real wireless communication networks. For example, due to the mobility in cellular networks, the number of active users during the delivery phase might not be the same as in the caching phase. %Therefore, a centralized caching and delivery policy will not be satisfactory in these networks. 
A ``decentralized`` approach of coded caching is proposed in \cite{DecentCodedCaching} where during the placement phase each user randomly and independently caches some parts of each file. Then, during the delivery phase, the server (or base station) takes the advantage of the multicast opportunities which has been created during the caching phase to satisfy the user requests.

Several works can be found in literature studying caching schemes for Device-to-Device (D2D) communication networks \cite{D2DVideo,OptThrouOneHop,FemtoCacheD2D,D2DCodedCaching,D2DCaching}. A one-hop D2D network is considered in \cite{OptThrouOneHop} where each user caches a subset of files randomly. It has been shown that the throughput per user varies as a function of the user cache size over the file library size under this scheme for sufficiently large number of users and files. %A caching network is considered in \cite{Femtocaching} where a number of helper nodes aim to serve the requests of the users in their coverage area from their cache memory content. All of the users and the helper nodes have also access to a base station who has access to the whole file library. %Uncoded and coded caching schemes are proposed and it has been shown that the uncoded cache placement is an NP-hard set cover problem while the coded cache placement is a convex problem which can be reduced to a linear program. 
The coded caching for D2D communication networks is studied in \cite{D2DCodedCaching}. In the proposed model, there exists no central server (or base station) in the network to control the caching and delivery phases while all users contribute by multicasting their useful cache content to other users during the delivery phase in order to satisfy any user request.  %A centralised and decentralised caching scheme has been proposed and achievable rates has been established.

%In all of the cooperative or device-to-device caching schemes in literature, it is assumed that either all of the users are cooperative such that during the delivery phase all of the users in the network will participate in transmitting their cache content which is useful for the other users in the network. This constraint somehow limits the applicability of these schemes in real wireless or cellular networks. For example, 
In cellular networks it is likely that some of the users have left the network during the delivery phase which makes their cache content not available for the other users during the delivery phase. Moreover, it is likely that some of the users will not participate during the delivery phase due to their privacy settings or their energy saving concerns.

In this paper, we consider a partially cooperative D2D communication network without a central server (or base station) where some of the users are selfish. More precisely,  although all of the users will cache the file library in their cache memory, only a group of non-selfish users participate in delivery phase by transmitting their useful cache content. Obviously, all of the users including selfish users can receive these transmissions to recover their requested file. We propose a ``deterministic'' and ``random'' caching scheme where we show that these schemes can satisfy any arbitrary user requests at the presence of selfish users in the network (i.e., partial cooperation of the users). We also establish the necessary conditions and achievable rates of these schemes.

The remainder of this paper is organised as follows. In Section \ref{sec:ProbState}, we introduce the network model and the mathematical tools and definitions which we will be using in our proposed caching schemes. In Sections \ref{sec:DetCache} and \ref{sec:RandCache} we present our deterministic and random caching schemes for D2D networks with selfish users, respectively. We discuss the caching and delivery procedures and provide examples, the corresponding necessary conditions and achievable rates. We then conclude the paper in Section \ref{sec:Concl}.

%%%%%%%%%%%%%%%%%%%%%%%%%%%%%%%%%%%%%%%%%%%%%%%%%%%%%%
\section{Problem Statement} \label{sec:ProbState}
%%%%%%%%%%%%%%%%%%%%%%%%%%%%%%%%%%%%%%%%%%%%%%%%%%%%%%
Our proposed network consists of $K$ users $\U=\{u_1,\ldots,u_K\}$ each with a cache memory of size $MB$ bits. The users of the network are divided into two groups of $\emph{selfish}$ and $\emph{non-selfish}$ users specified by the subsets $\S$ and $\T$, respectively, such that $\U=\S \cup \T$. All users are in the coverage range of each other such that they all can communicate with each other.

\begin{definition}[File Library]
The file library is a set of $N$ popular demanding files (e.g. videos) $\Omega=\{\omega_1,\ldots,\omega_N\}$ each of size $B$ bits.
\end{definition}

%Throughout this work we assume that all the files in the library have an i.i.d distribution (i.e., the files are generated independently and have the same popularity). Moreover, each user makes an arbitrary request from the file library independent from the other users which forms an arbitrary request vector $\R=(\sr_1,\sr_2,\ldots,\sr_K)$ such that $\sr_k \in \Omega$.

We assume that each user $u_i$ makes an arbitrary request $\sr_k$ from the file library $\Omega$ independent from the other users. Denote the request vector of all users by $\R=(\sr_1,\sr_2,\ldots,\sr_K)$, where $\sr_k \in \Omega$ for all $k \in \{1,\ldots,K\}$.

\begin{definition}[Caching Function]
The caching function consists of a function $\phi_u$ to map the file library into the cache memory of user $u$ such that $\phi_u: \FF_2^{NB} \rightarrow \FF_2^{MB}$.
\end{definition}

\begin{definition}[Delivery Functions]
The delivery functions consist of two functions $\psi_u: \FF_2^{MB} \times \Omega^K \rightarrow \FF_2^{R_uB}$ and $\Psi_u: \FF_2^{B\sum_{v \in \sT_u}R_v} \times \FF_2^{MB} \times \Omega^K \rightarrow \FF_2^B$, respectively, for generating the transmitted message of user $u$ and for decoding the received messages at user $u$, where $\sT_u$ is the set of users whose messages are received at user $u$ and $R_u$ is the transmission rate of user $u$.
\end{definition}

The caching and delivery functions are the essential tools to characterize the caching and delivery procedures as follows.

\noindent
$\textbf{Caching Phase:}$
During the caching phase, the cache content $\Gamma_u$ of each user $u$ is generated and stored in its cache memory by employing the caching function such that $\Gamma_u=\phi_u(\omega_n: \forall \omega_n \in \Omega)$.

\noindent
$\textbf{Delivery Phase:}$
During the delivery phase, each non-selfish user $u$ generates its transmitted message $X_{u,\R}=\psi_u(\Gamma_u,\R)$ as a function of its cache content $\Gamma_u$ and the user request vector $\R$. Furthermore, each user $u$ decodes its requested file $\hat{\omega}_{u,\R}=\Psi_u(\{X_{v,\R}: \forall v \in \sT_u\}, \Gamma_u, \R)$ as a function of the received messages $X_{v,\R}$ from all the transmitter nodes $v \in \sT_u$, its cache content $\Gamma_u$, and the request vector $\R$.

During the caching phase all of the users in $\U$ will generate their cache content as a function of the file library $\Omega$. However, only non-selfish users in $\T$ will participate during the delivery phase. Then, the worst-case error probability is given by
\[
\text{Pr}_\text{e}= \max_{\R} \max_{u \in \U} \text{P} \left(\hat{\omega}_{u,\R} \neq \omega_{u}\right).
\]

\begin{definition}[Achievable rate]
Let $R=\sum_{u \in \U}R_u$ where $R_u$ is the transmission rate of non-selfish user $u$ to satisfy a request vector $\R$. Then, the Memory-Rate pair $(M,R)$ (i.e., R(M)) is achievable if $R(M)$ is achievable for any arbitrary request vector $\R$. In other words, the rate $R(M)$ is said to be achievable if for any $\epsilon > 0$ and sufficiently large file size $B$ there exists a caching-delivery scheme such that $\text{\emph{Pr}}_\text{\emph{e}} \leq \epsilon$.
\end{definition}

%%%%%%%%%%%%%%%%%%%%%%%%%%%%%%%%%%%%%%%%%%%%%%%%%%%%%%
\section{Deterministic Caching} \label{sec:DetCache}
%%%%%%%%%%%%%%%%%%%%%%%%%%%%%%%%%%%%%%%%%%%%%%%%%%%%%%
In this section we propose a deterministic caching scheme where each user will cache a predetermined parts of each file. As it is proposed in \cite{CodedCaching} and \cite{D2DCodedCaching}, deterministic caching procedure will create a multicasting opportunity among all users during the delivery phase in order to retrieve their requested file. As mentioned earlier, only non-selfish users will transmit their cache content during the delivery phase. Let $t=\frac{MK}{N}$ be an integer between 1 and $K-1$. Note that when $t \geq K$, each user can cache all the $N$ files.

\noindent $\bullet$ {\textbf{Caching phase:}} The caching procedure consists of the following steps:
\begin{enumerate}
\item
Each file $\omega_n \in \Omega$ is divided into $t{K \choose t}$ subfiles.
\item
The subfiles of each file are partitioned into ${K \choose t}$ groups, each of which has $t$ subfiles. Each group is indexed by a specific subset of $t$ users.
\item
For every file in the library, its subfiles in each group will be cached in the corresponding subset of users. For example, the subfiles in the group with index $\{u_{i_1}, u_{i_2}, \ldots u_{i_t}\}$ will be cached in users $u_{i_1}, u_{i_2}, \ldots, u_{i_t}$.
\end{enumerate}

This caching procedure guarantees that any subset of $t$ users share the same $t$ subfiles of each file in the library. Furthermore, each user caches a total of ${K-1 \choose t-1} tN$ subfiles. Since each subfile is of size $B / [t {K \choose t}]$ bits, the total number of bits stored in the cache of each user is $tNB/K = MB$, which satisfies the cache memory constraint.

\begin{example} \label{exmp:DetCachePlace}
Consider a network of $K=4$ users each of which has cache size $M=3$ and a file library $\Omega=\{\omega_1,\omega_2,\omega_3,\omega_4\}$ of size $N=4$ files. Since $t=3$, then each file are divided into $12$ subfiles and any subset of $t=3$ users will cache the same $3$ subfiles from each file. The cache content of each user is depicted in Figure \ref{fig:DetCache} which follows from the aforementioned caching phase.

\begin{figure}[htb]
\centering
\includegraphics[width=.48\textwidth]{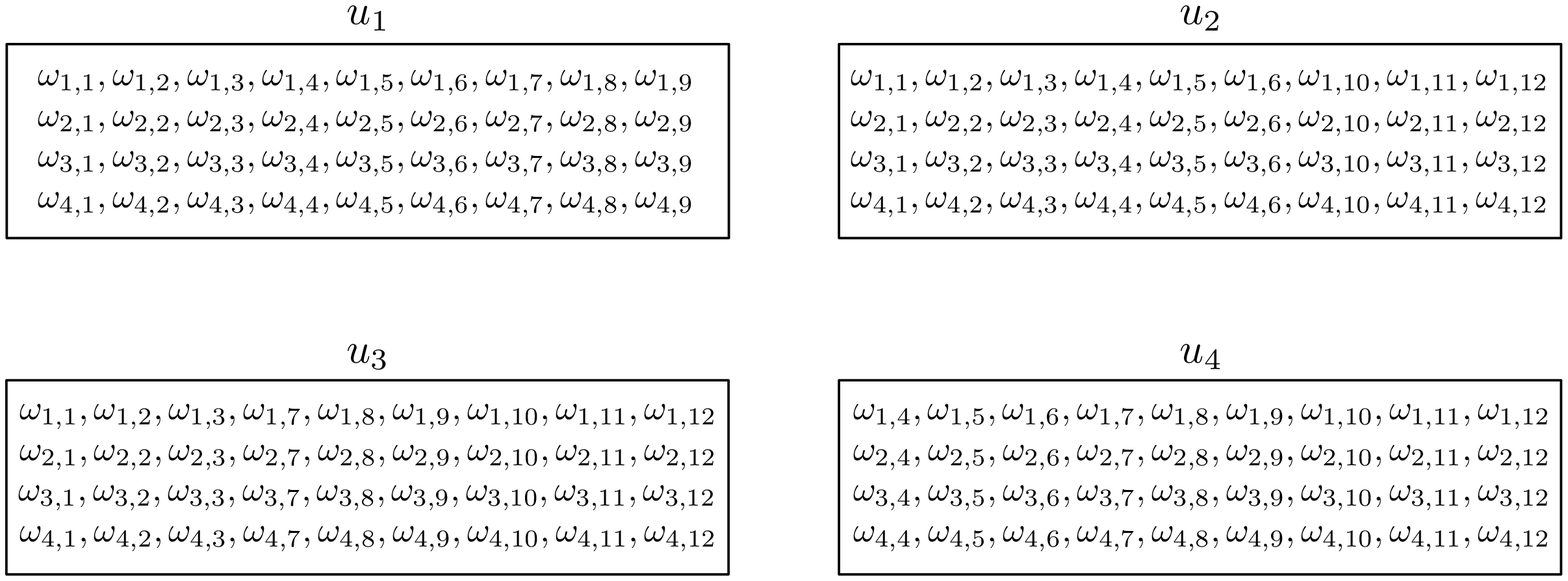}
\caption{The users cache memory content of a network with $K=4$ users of cache size $M=3$ and a file library of size $N=4$ with deterministic caching.}
\label{fig:DetCache}
\end{figure}
\end{example}

In contrast to the delivery phases proposed in \cite{DecentCodedCaching} and \cite{D2DCodedCaching} where a base station has access to the whole file library and is responsible for all the transmission or a fully cooperative users during the delivery phase, in our proposed scheme, a group of selfish users will not participate during the delivery phase in order to transmit their useful cache contents. Therefore, the non-selfish users have to compensate for the selfishness of the selfish nodes through transmitting more useful information to other users. Our scheme can tolerate at most $t-1$ selfish users in the network.

\noindent $\bullet$ {\textbf{Delivery phase:}} Let $\sU \subset \U$ be a subset of $t+1$ users. Furthermore, let $\sU= \S \cup \T$, where $\S$ and $\T$ are the sets of selfish and non-selfish users in $\sU$, and $|\S| \leq t-1$. Consider a user $u_k \in \sU$. The subset of users $\sU \setminus \{u_k\}$ shares $t$ subfiles that are useful for $u_k$ in order to retrieve its requested file. Denote the set of these $t$ subfiles by $\mathcal{P}_{\sU,k}$. If user $u_k$ obtained all subfiles in $\mathcal{P}_{\sU,k}$, for all $\sU \subset \U$ with $u_k \in \sU$ and $|\sU| = t+1$, then it would have a total of ${K-1 \choose t}t$ useful subfiles. Since there are ${K-1 \choose t-1}t$ subfiles of its requested file in its own cache memory, user $u_k$ would have obtained its whole requested file. In other words, user $u_k$ needs to obtain one distinct subfile from each user in $\sU \setminus \{u_k\}$. This implies that each user $u_i \in \sU$ can provide one useful subfile to each user in $\sU \setminus \{u_i\}$, and these $t$ useful subfiles that user $u_i$ can deliver are distinct.

We consider the following cases for all $\sU \subset \U$ of size $|\sU|=t+1$:
\begin{enumerate}
\item First, $\S = \emptyset$. Then, $\T=\sU$. In this case, to deliver the $t$ useful subfiles to each user in $\sU$, each user $u_i \in \T$ will multicast the XOR of these $t$ subfiles which are useful for the remaining $t$ users in $\sU \setminus \{u_i\}$  by transmitting
\[
\bigoplus_{u_k \in \sU \setminus \{u_i\}} \Gamma_{\sr_k,\sU \setminus \{u_k\}},
\]
where, $\Gamma_{\sr_k,\sU \setminus \{u_k\}}$ is a subfile of the file $\sr_k$ which is requested by $u_k$ and is stored in the cache memory of all users in $\sU \setminus \{u_k\}$ but is missing in $u_k$ . Each user in $\sU \setminus \{u_i\}$ can extract its required subfile using the side information in its cache.

\item Next, $\S \neq \emptyset$ and $|\S| < |\T|$. For each non-selfish user in $\T$, it just broadcasts the same packet as in Case~1. For each selfish user $u_s \in \S$, we pick an arbitrary user $u_t \in \T$ to transmit on behalf of $u_s$. According to the delivery procedure in Case~1, $u_s$ should transmit a packet in the form of
    $$
    \bigoplus_{u_k \in \sU \setminus \{u_s\}} \Gamma_{r_k,\sU \setminus \{u_k\}},
    $$
    where $\Gamma_{\sr_k,\sU \setminus \{u_k\}}$ is a subfile of file $\sr_k$ requested by user $u_k$ which is cached in all users in $\sU \setminus \{u_k\}$ but is missing in user $u_k$'s cache memory. Since $u_s$ is selfish and does not participate in the delivery process, a user $u_t$ in $\T$ is picked to transmit
    $$
    \bigoplus_{u_k \in \sU \setminus \{u_s,u_t\}} \Gamma_{r_k,\sU \setminus \{u_k\}},
    $$
    which means that user $u_t$ transmits the XOR of the same subfiles as $u_s$ except the one which is useful for user $u_t$ itself. In order to compensate for all the selfish users, a set of $|\S|$ arbitrary helpers are picked from $\T$, which is denoted by $\HH$. Each helper in $\HH$ transmits in the above way on behalf of a user in $\S$. Note that each helper misses one subfile from its requested file which is stored in the cache memory of the selfish user it transmits on behalf of. In order to compensate for these missing subfiles, a user is picked from $\T \setminus \HH$ to transmit the XOR of these $|\S|$ missing subfiles. Note that such a user can always be found since $|\S| < |\T|$. Each user in $\HH$ can then extract its missing subfile using the side information in its cache.

\item Last, $\S \neq \emptyset$ and $|\S| \geq |\T|$. For each non-selfish user in $\T$, it just broadcasts the same packet as in Case~1. By assumption, $|\sU| = t+1$ and $|\S| \leq t-1$, which implies $|\T| \geq 2$. We let all users in $\T$ be the helpers of $|\T|-1$ selfish users in $\S$. The delivery procedure is the same as in Case~2. Afterwards, we remove the $|\T|-1$ selfish users from $\S$ to obtain $\S'$. Either the procedure in Case~2 or in Case~3 is repeated depending on whether $|\S'|$ is smaller than $|\T|$ or not. The delivery procedure ends when the transmissions of all selfish users have been compensated for.
\end{enumerate}

\begin{example}
Consider the network in Example \ref{exmp:DetCachePlace} where $S=2$ of the users are selfish. Let the request vector be $\R=(\omega_1,\omega_2,\omega_3,\omega_4)$. There exists only one subset of $t+1=4$ users which is $\sU=\{u_1,u_2,u_3,u_4\}$. Without loss of generality, let $\S=\{u_1,u_2\}$ be the set of selfish users. Therefore, $|\S|=|\T|$ where the delivery procedure lays on  Case~3. Therefore, non-selfish users $\T=\{u_3,u_4\}$ multicast $\omega_{1,11} \oplus \omega_{2,8} \oplus \omega_{4,3}$ and $\omega_{1,12} \oplus \omega_{2,9} \oplus \omega_{3,6}$ which are useful for users $\{u_1,u_2,u_4\}$ and $\{u_1,u_2,u_3\}$, respectively. In order to compensate for selfish user $u_1$, user $u_3$ multicasts $\omega_{2,7} \oplus \omega_{4,1}$ which is useful for users $\{u_2,u_4\}$ and $u_4$ unicasts $\omega_{3,4}$ which is useful for user $u_3$. Moreover, in order to compensate for selfish user $u_2$, user $u_3$ multicasts $\omega_{1,10} \oplus \omega_{4,2}$ which is useful for users $\{u_1,u_4\}$ and $u_4$ unicasts $\omega_{3,5}$ which is useful for user $u_3$.
\end{example}

%\begin{remark}
%According to the delivery phase, at least two non-selfish users are required to compensate for one selfish user in any subset of $t+1$ users. Therefore, our proposed caching scheme is able to tolerate at most $t-1$ selfish user in the network.
%\end{remark}

\begin{theorem}
Let $M$, $N$, $K$, and $S$ be the cache size of each user, number of the popular files, number of users, and number of the selfish users in the network, respectively. For $t=\frac{MK}{N} \in \ZZ^+$, the following rate is achievable:
\begin{align}
R(M)=\frac{1}{t { K \choose t}} \sum_{i=0}^{S} {S \choose i} {K-S \choose t+1-i} \Big( t+1 + \left \lceil \frac{i}{t-i} \right \rceil  \Big)
\end{align}
\end{theorem}

\begin{proof}
The number of $(t+1)$-subsets of $\sU$ which contain $i$ selfish users is given by ${S \choose i} {K-S \choose t+1-i}$. Consider the number of transmissions in such a subset. Since each non-selfish user transmits according to Case~1, the $t+1-i$ non-selfish users will first transmit a total of $t+1-i$ packets. Next, they need to compensate for the non-participation of the selfish users. Let $q$ and $r$ be the quotient and remainder when dividing $i$ by $t-i$. Then the rule in Case~3 will be repeatedly applied for $q$ times, which require a total of $q(t+1-i)$ transmissions. If $r = 0$, the delivery phase ends; otherwise, the rule in Case~2 will be applied, which requires $r+1$ transmissions. Let $I(r>0)$ be the indicator function which equals 1 if $r>0$ and equals 0 otherwise. Consequently, the total number of transmissions by this subset of users is
\begin{align*}
& \big[ t+1-i \big] + \big[ q(t+1-i) + r + I(r>0) \big] \\
= & \; t + 1 + q + I(r>0) \\
= & \; t+1 + \left \lceil \frac{i}{t-i} \right \rceil 
\end{align*} 
Since the length of each transmission is $\frac{1}{t {K \choose t}}$ of the file size, the statement then follows.
\end{proof}

%\begin{proof}
%Without loss of generality, let $\{1,2,\ldots,S\}$ and $\{S+1,S+2,\ldots,K\}$ be the set of selfish and transmitter users, respectively. Any subset of $t+1$ out of $K-S$ transmitter users corresponds to $t+1$ transmissions. Moreover, any subset of $t+1$ out of $K$ users which contains $i$ selfish nodes, for $i=1,\ldots,S$ corresponds to
%\begin{enumerate}
%\item $t+1-i$ transmissions each of which useful for $t$ other users in the subset. Plus
%\item $i$ transmissions each of witch useful for $t-1$ other users in the subset. Plus
%\item $i$ transmissions each of witch useful for only one user in the subset.
%\end{enumerate}
%Therefore, the total number of transmissions corresponds to such subsets are $t+1+i$. Since the length of each transmission is $\frac{F}{t {K \choose t}}$, then the theorem follows.
%\end{proof}

%%%%%%%%%%%%%%%%%%%%%%%%%%%%%%%%%%%%%%%%%%%%%%%%%%%%%%
\section{Random Caching} \label{sec:RandCache}
%%%%%%%%%%%%%%%%%%%%%%%%%%%%%%%%%%%%%%%%%%%%%%%%%%%%%%
Due to the dynamic nature of the wireless networks, the deterministic caching policy may not be able to fully advantage D2D networks proposed in this paper. The deterministic caching approach, as explained in Section \ref{sec:DetCache}, requires full control on the caching phase since some predetermined subfiles of each file must be placed in the cache memory of each user. In a wireless D2D network, however, devices may leave or join the network at any time. Therefore, a decentralized or random caching policy is more realistic and practical as it fits better these networks characteristics.

In contrast to the decentralized caching algorithm proposed in \cite{DecentCodedCaching}, randomly caching the bits of each file in D2D networks does not guarantee a successful recovery of the requested files due to the lack of a server (or base station) during the delivery phase to provide the subfiles that are not cached in any user. In order to overcome this issue, we follow the caching procedure in \cite{D2DCodedCaching}, where an MDS-code is employed to encode the library files before caching by the users. Consequently, any user can retrieve its requested file by having access to a subset of sufficient encoded symbols.

\noindent $\bullet$  ${\textbf{Caching phase:}}$ Every file $\omega_1,\ldots,\omega_N$ of length $B$ bits is divided into $I$ subfiles each of size $\frac{B}{I}$ bits. Each subfile is considered as a symbol in $\FF_{2^{B/I}}$. Then, each file is encoded using an $(I,\frac{I}{r})$ MDS-code over $\FF_{2^{B/I}}$, where $r$ is the appropriate code rate which we will establish its necessary condition. Each user $u \in \U$ then randomly selects a set of $\frac{MI}{N}$ indices which indicate the corresponding encoded symbols of each file to be cached in the memory of user $u$.

By the end of the caching phase, the encoded symbols of any file $\omega_i$ for $i=1,\ldots,N$ can be partitioned as $\omega_{i,\P}$ where $\P \subset \{1,\ldots,K\}$ such that $\omega_{i,\P}$ identifies the encoded symbols of file $\omega_i$ cached {\em exclusively} in the users of set $\P$. For instance, for a network with $N=4$ files, $\omega_{1,\{2,4\}}$ specifies the encoded symbols of file $\omega_1$ cached in users $u_2$ and $u_4$, but {\em not} cached in any other user. Similarly, $\omega_{1,\emptyset}$ specifies the encoded symbols of file $\omega_1$ which are cached in no user. Let $\Gamma_{\omega_{i,\P}}$ be the block (i.e., concatenation) of encoded symbols of file $\omega_i$ specified by $\omega_{i,\P}$. As a consequence, the cache memory content of any user $u \in \U$ can be presented as $\Gamma_{\omega_{1,\P}},\Gamma_{\omega_{2,\P}},\ldots,\Gamma_{\omega_{N,\P}}$ for all $\P \subset \U$ and $u \in \P$.

During the caching phase, each user caches $\frac{MI}{N}$ out of $\frac{I}{r}$ encoded symbols of each file. Therefore, the probability that a given encoded symbol of any file is cached in an arbitrary user is given by $Mr/N$.
%\begin{align*}
%P(\text{each encoded symbol of any file is cached in each user})  =\frac{Mr}{N}.
%\end{align*}
According to the Law of Large Numbers, the number of encoded symbols of any file $\omega_i$ exclusively cached by a set of users $\P$ is given by
\begin{align} \label{eq:LLN}
|\omega_{i,\P}| \approx (\frac{Mr}{N})^{|\P|}(1-\frac{Mr}{N})^{K-|\P|}\frac{I}{r},
\end{align}
with a high probability for large number of symbols $I$.

The delivery phase we are going to propose for D2D networks with selfish users guarantees that each user $u$ at the end of the procedure has access to the encoded symbols of its requested file cached in all non-selfish users. Obviously, if there exist no selfish user in the network, each user can have access to all encoded symbols of its requested file which are cached in all other users \cite{D2DCodedCaching}.

\begin{example} \label{exmp:1}
Consider a network with a file library of $N=4$ files and $K=4$ users each of which has a cache size $M=2$. Assume user $u_1$ requests file $\omega_1$. As mentioned earlier, the cached symbols of encoded file $\omega_1$ in user $u_1$ can be presented as $\Gamma_{\omega_{1,\{1\}}}$, $\Gamma_{\omega_{1,\{1,2\}}}$, $\Gamma_{\omega_{1,\{1,3\}}}$, $\Gamma_{\omega_{1,\{1,4\}}}$, $\Gamma_{\omega_{1,\{1,2,3\}}}$, $\Gamma_{\omega_{1,\{1,2,4\}}}$, $\Gamma_{\omega_{1,\{1,3,4\}}}$, and $\Gamma_{\omega_{1,\{1,2,3,4\}}}$. At the end of the delivery phase, user $u_1$ will also have access to the encoded symbols of file $\omega_1$ cached in all of the other users (due to the fully cooperative network), namely $\Gamma_{\omega_{1,\{2\}}}$, $\Gamma_{\omega_{1,\{3\}}}$, $\Gamma_{\omega_{1,\{4\}}}$, $\Gamma_{\omega_{1,\{2,3\}}}$, $\Gamma_{\omega_{1,\{2,4\}}}$, $\Gamma_{\omega_{1,\{3,4\}}}$, and $\Gamma_{\omega_{1,\{2,3,4\}}}$. In other words, all of the encoded symbols of file $\omega_1$ cached in all subsets $\P \subset \{1,2,3,4\}$ of the users are known by user $u_1$ at the end of the delivery phase.
%%%%%%%%%%%%
\end{example}

%Example \ref{exmp:1} can be easily generalized for a network with $K$ users each of cache size $M$ and a file library of size $N$.
%%%%%%%%%%%%%%%%%%%%%%%

\subsection{Rate of the MDS Code}

Now we initiate a discussion on the rate of the MDS-code which is used in the caching phase.

\begin{proposition} \label{prop:TotalSymb}
For a network with $K$ users each of cache size $M$ and a file library of size $N$, the number of the known encoded symbols of any requested file at each user by the end of the delivery phase is
\begin{align} \label{eq:KnownSymb}
\frac{I}{r}\sum_{i=1}^{K}{K \choose i}\left(\frac{Mr}{N}\right)^i \left(1-\frac{Mr}{N}\right)^{K-i}.
\end{align}
\end{proposition}

% As mentioned earlier, Proposition \ref{prop:TotalSymb} follows from the fact that each user has access to all encoded symbols of its requested file cached in all subset of users $\P \subset \U$.

\begin{proof}
Each user has access to all encoded symbols of its requested file cached in all subset of users $\P \subset \U$. The statement follows from~\eqref{eq:LLN} and the fact that there are ${K \choose i}$ subsets of size $i$.
\end{proof}

\begin{theorem} \label{thm:CodeRate}
For a network with $K$ users each of cache size $M$ and a file library of size $N$ where $\frac{MK}{N}>1$, any arbitrary request vector $\R=(\sr_1,\sr_2,\dots,\sr_K)$ is recoverable by the users if all the files of the library are encoded using an MDS-code with rate $r$ such that $r<1$ is the real positive root of the polynomial
\begin{align} \label{eq:CodeRate}
f(r) \triangleq \sum_{i=0}^{K-1}{K \choose i} \left(-\frac{M}{N}\right)^{K-i}r^{K-i-1}+1=0.
\end{align}
\end{theorem}

\begin{proof}
Since the $I$ symbols (i.e., subfiles) of each file are encoded by an $(I,\frac{I}{r})$ MDS-code, any $I$ symbols of the encoded file is sufficient to recover the requested file. Therefore, the the number of symbols of the requested file known by the requesting user in \eqref{eq:KnownSymb} must satisfy
\begin{align} \label{eq:KnownSymbI}
\frac{I}{r}\sum_{i=1}^{K}{K \choose i}\left(\frac{Mr}{N}\right)^i \left(1-\frac{Mr}{N}\right)^{K-i} \geq I,
\end{align}
which can be re-written as
\begin{equation}
\frac{1}{r} \left[(1-\frac{Mr}{N})^K - 1 \right] + 1 = f(r) \leq 0.
\end{equation}
Since $f(1) > 0$ and $f(0) = 1 - \frac{MK}{N} < 0$, the polynomial equation $f(r)=0$ must have a real root $r \in (0,1)$.
%Moreover, the polynomial in \eqref{eq:KnownSymbI} has a real positive root $r \in (0,1)$ if the constant term of the polynomial satisfies $\frac{MK}{N}-1 > 0$.
\end{proof}

Note that the condition $\frac{MK}{N}>1$ is required so that the total memory size of the $K$ users is larger than the size of the whole library of files.

\begin{example}
Consider the same network as in Example \ref{exmp:1}. Then, the polynomial in \eqref{eq:CodeRate} is given by
\[
(1/2)^4 r^3-4(1/2)^3r^2+6(1/2)^2r-4(1/2)+1=0.
\]
The real positive root of the polynomial is $r=0.91$ which specifies the rate of the MDS-code to encode the library files. Encoding the library files with this rate guarantees that any arbitrary request vector is recoverable by the users.
\end{example}

\begin{theorem} \label{thm:SelfishCodeRate}
For a network with $K$ users each with a cache of size $M$ such that $S$ out of $K$ users are selfish, and a file library of size $N$ where $\frac{M(K-S)}{N}>1$, any arbitrary request vector $\R=(\sr_1,\sr_2,\dots,\sr_K)$ is recoverable by the users if all the files of the library are encoded using an MDS-code with rate $r$ such that $r<1$ is the real positive root of the polynomial
\begin{align} \label{eq:SelfishCodeRate}
\sum_{i=0}^{K-S-1}{K-S \choose i} \left(-\frac{M}{N}\right)^{K-S-i}r^{K-S-1-i}+1=0.
\end{align}
\end{theorem}

\begin{proof}
Note that Theorem~\ref{thm:CodeRate} is true because each user is able to access the cache of $K-1$ other users at the end of the delivery phase. When $S$ out of the $K$ users are selfish, each selfish user can access the cache of $K-S$ non-selfish users while each
non-selfish user can access the cache of $K-S-1$ other non-selfish users, so the request of a non-selfish user is harder to satisfy. The statement then follows directly from Theorem~\ref{thm:CodeRate} by replacing $K$ by $K-S$.
\end{proof}

%%%%%%%%%%%%%
\subsection{Delivery phase and the achievable rate}
%%%%%%%%%%%%%
Now, we propose the delivery algorithm of our D2D random caching scheme where a subset of users in the network are selfish. Before describing the procedure, we need to define some notations first. Let $\Gamma_{\omega_i,\P} \oplus \Gamma_{\omega_i,\P'}$ be the XOR of these two symbol blocks. In general, these two symbol blocks may not be of equal length. In that case, we zero-pad the shorter one so that the two blocks are of the same length. The notation can be generalized to more than two symbol blocks. Besides, in our delivery procedure, a symbol block is usually divided into a certain number, say $f$, of disjoint, equal-length segments. We call such a segment an $\frac{1}{f}$-segment of a symbol block. Consider the case where there are several users, each of which has a symbol block. These symbol blocks are of the same length but their contents can be different. Each user transmits an $\frac{1}{f}$-segment of its own symbol block. If the offsets of these $\frac{1}{f}$-segments are all different, then we say that those segments are {\em disjoint}.

\noindent $\bullet$ {\textbf{Delivery phase:}} In each iteration of the delivery procedure, a subset $\sU \subset \U$ of size $|\sU|=\su$ for all $\su=K,K-1,\ldots,2$ is selected. Let $\S \subset \sU$ and $\T \subset \sU$  be the subsets of selfish and non-selfish users in $\sU$, and let $|\T| = \st$. Do the following for each $\sU$ that contains one or more non-selfish users:
\begin{enumerate}
\item If $|\T| = 1$, then let the only user $u\in \T$ transmits $\bigoplus_{v \in \sU, v \neq u} \Gamma_{\omega_v,\sU \setminus \{v\}}$.
\item If $|\T|\geq 2$, then
    \begin{itemize}
     \item Pick an arbitrary user $u^*$ from $\T$. User $u^*$ transmits a disjoint $\frac{1}{\st-1}$-segment of $\bigoplus_{u \in \T, u \neq u^*} \Gamma_{\omega_u,\sU \setminus \{u\}}$. \\
     \item Each user $u \in \T \setminus\{u^*\}$ transmits a disjoint $\frac{1}{\st-1}$-segment of $\bigoplus_{v \in \sU, v \neq u} \Gamma_{\omega_v,\sU \setminus \{v\}}$.
     \end{itemize}
\end{enumerate}

Now, we establish the achievable rate based on the delivery algorithm proposed earlier.

\begin{theorem} \label{thm:SelfishTransRate}
For a network with $K$ users each of cache size $M$ such that $S$ out of $K$ users are selfish, and a file library of size $N$ which are encoded by a MDS-code of rate $r$ satisfying the condition in Theorem \ref{thm:SelfishCodeRate}, the following rate is achievable:
\begin{align} \nonumber
R(M)=\frac{1}{r}\sum_{i=2}^{K} \sR(i) \left(\frac{Mr}{N}\right)^{i-1} \left(1-\frac{Mr}{N}\right)^{K-i+1}
\end{align}
where $\sR(i)$ is given by
\begin{align} \label{eq:RandTrans}
\sR(i) = \sum_{j=0}^{i-2} {S \choose j} {K-S \choose i-j} \frac{i-j}{i-j-1}+(K-S) {S \choose i-1}.
\end{align}
\end{theorem}

\begin{proof}
The function $\sR(i)$ in \eqref{eq:RandTrans} counts the number of transmissions for each subset of $i$ users. In the first term of \eqref{eq:RandTrans}, ${S \choose j} {K-S \choose i-j}$ is the number of the subsets of $i$ users which contain $i-j \geq 1$ non-selfish users. Each of such subsets transmits $\frac{1}{i-j-1}$-segment of $i-j$ packets according to the Case~2 of the delivery phase. In the second term of \eqref{eq:RandTrans}, $(K-S) {S \choose i-1}$ is the number of the subsets of $i$ users which contain only one non-selfish user and each of which transmits one packet according to the Case~1 of the delivery phase. Moreover, the number of the encoded symbols which are transmitted by each non-selfish user in a subset of $i$ users is $\frac{I}{r} \left(\frac{Mr}{N}\right)^{i-1} \left(1-\frac{Mr}{N}\right)^{K-i+1}$. Since the size of each symbol is $\frac{B}{I}$ bits, then the theorem follows.
\end{proof}

\begin{example}
Consider a network with $K=5$ users each of which has a cache memory of size $M=2$ while $S=2$ users are selfish, and a library of size $N=5$ files. In order to satisfy any arbitrary request, all files must be encoded by a MDS-code of rate $r=0.44$ follows from \eqref{eq:SelfishCodeRate} in Theorem \ref{thm:SelfishCodeRate}. Assume, without loss of generality that users $u_1$ and $u_2$ are selfish and the request vector is $\R=(\omega_1,\omega_2,\omega_3,\omega_4,\omega_5)$. Then, the delivery phase will be carried out for all subsets $\sU \subset \U$ of size $|\sU|=5,4,3,2$. Suppose, $\sU=\{1,2,3,4,5\}$ which contains three non-selfish users $\T=\{3,4,5\}$. Following the Case~2 of the delivery phase, we choose $u^*=u_3$. Then, $u_3$ will transmit $\frac{1}{2} \Gamma_{\omega_4,\{1,2,3,5\}} \oplus \Gamma_{\omega_5,\{1,2,3,4\}}$. Moreover, users $u_4$ and $u_5$ will respectively transmit $\frac{1}{2} \Gamma_{\omega_1,\{2,3,4,5\}} \oplus \Gamma_{\omega_2,\{1,3,4,5\}} \oplus \Gamma_{\omega_3,\{1,2,4,5\}} \oplus\Gamma_{\omega_5,\{1,2,3,4\}}$ and $\frac{1}{2} \Gamma_{\omega_1,\{2,3,4,5\}} \oplus \Gamma_{\omega_2,\{1,3,4,5\}} \oplus \Gamma_{\omega_3,\{1,2,4,5\}} \oplus\Gamma_{\omega_4,\{1,2,3,5\}}$. Now, let $\sU=\{u_1,u_2,u_3\}$ which contains only one non-selfish user $\T=\{u_3\}$. Following the Case~1 of the delivery phase, user $u_3$ will transmit $\Gamma_{\omega_1,\{2,3\}} \oplus \Gamma_{\omega_2,\{1,3\}}$. The same procedure will be carried out for all other $\sU$. At the end of the delivery phase, all users will have access to all encoded symbols cached in the network except the ones cached in selfish users. The transmission rate for this network setting is $R(2) = 3.05$ following from Theorem \ref{thm:SelfishTransRate}.
\end{example}

The performance of the ``deterministic`` and ``random`` caching schemes in a partially cooperative D2D network is depicted in Figure \ref{fig:Sim} for two different network settings. Figure \ref{fig:Sima} and \ref{fig:Simb} compare the achievable rates of deterministic and random caching schemes for a network with parameters $K=100$ users, $S=10$ selfish users, $N=50$ files and $K=50$ users, $S=10$ selfish users, $N=100$ files, respectively. The simulation results show that although the deterministic scheme has a slightly better performance (i.e., smaller transmission rate), the gap vanishes as the cache memory size $M$ increases. The lower bounds depicted in Figure \ref{fig:Sim} is based on the cut-set argument \cite{D2DCodedCaching} and is given by
\[
R^*(M) \geq \max_{l=\{1,2,\ldots,\min\{K,N\}\}} \left(l - \frac{l}{\lfloor{\frac{m}{l}}\rfloor}M\right),
\] 
where, $R^*(M)$ is the optimal achievable rate defined as 
\[
R^*(M) \triangleq \inf{\{R(M) : \: R(M) \text{~is achievable}\}}.
\]

\begin{figure}[ht] 
  \centering
  \begin{subfigure}{\linewidth}
    \centering
    \includegraphics[width=1.02\linewidth]{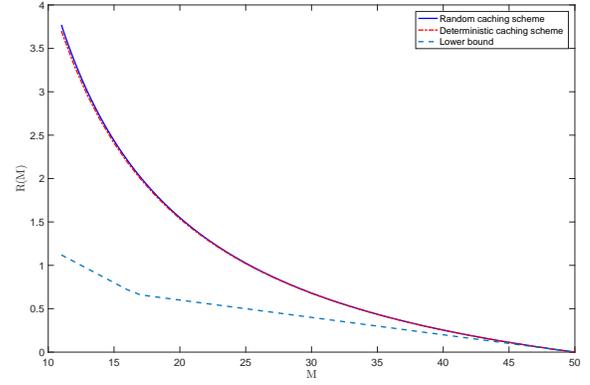}
    \caption{$K=100$, $N=50$, $S=20$.}
    \label{fig:Sima}
  \end{subfigure}

  \begin{subfigure}{\linewidth}
    \centering
    \includegraphics[width=1.02\linewidth]{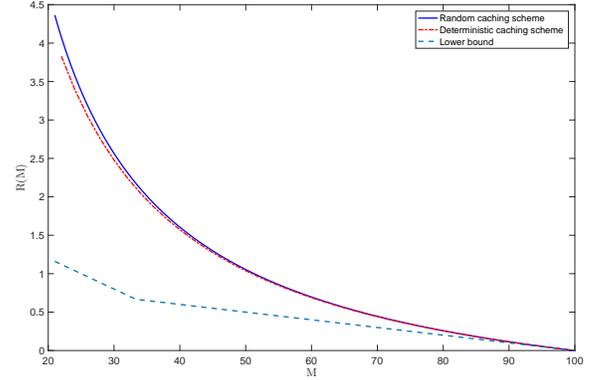}
    \caption{$K=50$, $N=100$, $S=10$.}
    \label{fig:Simb}
  \end{subfigure}  
  \caption{Performance of the deterministic and random caching schemes in a partially cooperative D2D networks with $S$ selfish users out of $K$ users and a file library of size $N$ files.}  
  \label{fig:Sim}
\end{figure}  

%%%%%%%%%%%%%%%%%%%%%%%%%%%%%%%%%%%%%%%%%%%%%%%%%%%%%%
\section{Conclusions} \label{sec:Concl}
%%%%%%%%%%%%%%%%%%%%%%%%%%%%%%%%%%%%%%%%%%%%%%%%%%%%%%
In this paper, we consider partially cooperative D2D caching networks where a group of users are selfish. All users cache the file library during the caching phase while only non-selfish users transmit their useful cache content to other users. We propose a deterministic caching scheme where the users' cache content is a deterministic function of the file library. We also propose a more realistic decentralized caching scheme where users randomly cache some parts of the files. We show that any arbitrary user request is recoverable under our proposed caching schemes at the presence of selfish users in the D2D networks.

%In this paper, we assume that all of the library files have the same popularity (i.e., uniform demands). A more realistic model would be a nonuniform distribution for user demands. Employing the same caching procedures proposed in this paper would not be an optimal method since they ignore the different popularity of the files. Moreover, Employing the same caching
%%%%%%%%%%%%%%%%%%%%%%%%%%%%%%%%%%%%%%%%%%%%%%%%%%%%%%

%%%%%%%%%%%%%%%%%%%%%%%%%%%%%%%%%%%%%%%%%%%%%%%%%%%%%%
% use section* for acknowledgement
%\section*{Acknowledgment}

%The authors would like to thank...
%%%%%%%%%%%%%%%%%%%%%%%%%%%%%%%%%%%%%%%%%%%%%%%%%%%%%%

\bibliographystyle{IEEEtran}
\bibliography{Ali_bib-1}

\end{document}